\documentclass[11pt]{article}
\usepackage{fullpage}

\usepackage{latexsym}
\usepackage{amsmath}
\usepackage{amssymb}
\usepackage{amsthm}
\usepackage{hyperref}
\usepackage{cite}
\usepackage{graphicx}
\usepackage{color}
\usepackage{amstext}
\usepackage{array}
\usepackage{xspace}
\usepackage{calc}
\usepackage{tikz}
\usetikzlibrary{decorations.markings}
\tikzstyle{vertex}=[circle, draw, inner sep=0pt, minimum size=6pt]

\newtheorem{theorem}{Theorem}
\newtheorem{theorem*}{Theorem}
\newtheorem{corollary}[theorem]{Corollary}
\newtheorem{lemma}[theorem]{Lemma}
\newtheorem{observation}[theorem]{Observation}
\newtheorem{proposition}[theorem]{Proposition}

\newtheorem{fact}[theorem]{Fact}

\theoremstyle{definition}

\newtheorem{example}[theorem]{Example}

\renewcommand{\vec}[1]{\mathbf{#1}}
\newcommand{\SL}{\mathrm{SL}}

\newcommand{\F}{\mathbb{F}}

\newcommand{\Z}{\mathbb{Z}}
\newcommand{\Q}{\mathbb{Q}}
\newcommand{\C}{\mathbb{C}}
\newcommand{\N}{\mathbb{N}}
\newcommand{\M}{\mathrm{M}}

\newcommand{\Tr}{\mathrm{Tr}}

\newcommand{\alt}{\mathrm{Alt}}
\newcommand{\sgn}{\mathrm{sgn}}
\newcommand{\ctabc}{\mathfrak{C}}
\newcommand{\ctabd}{\mathfrak{D}}
\newcommand{\ctabs}{\mathfrak{S}}

\newcommand{\ctabm}{\mathfrak{M}}

\newcommand{\po}[1]{|#1|}
\newcommand{\por}[1]{\parallel #1\parallel}
\newcommand{\pair}[1]{\widehat{#1}}
\newcommand{\poset}{\mathrm{P}}
\newcommand{\kpo}{\mathcal{K}}
\newcommand{\posetring}{\mathcal{P}}
\newcommand{\sg}{\mathrm{S}}
\newcommand{\specht}{\mathrm{Sp}}
\newcommand{\conjpart}[1]{\widetilde{#1}}
\newcommand{\colgp}{\mathrm{cl}}
\newcommand{\colsym}{\mathrm{c}}
\newcommand{\ipo}{\mathcal{E}}
\newcommand{\ctab}{\mathcal{C}}
\newcommand{\diagpart}{\delta}
\newcommand{\rowtab}{\mathcal{R}}
\newcommand{\mdisc}{\mathrm{mdisc}}

\newcommand{\vertex}{\node[vertex]}

\title{On generating the ring of matrix semi-invariants
}

\author{
G\'abor Ivanyos\thanks{Institute for Computer Science and Control, Hungarian 
Academy of Sciences, 
Budapest, Hungary. 
({\tt Gabor.Ivanyos@sztaki.mta.hu}).} 
\and
 Youming Qiao\thanks{Centre for Quantum Computation and Intelligent Systems,
 University of Technology, Sydney, Australia. 
 ({\tt jimmyqiao86@gmail.com})}
\and
K. V. Subrahmanyam \thanks{Chennai Mathematical Institute, Chennai, India.
({\tt kv@cmi.ac.in}).}
 }
\date{\today}

\begin{document}
\maketitle

\begin{abstract}
For a field $\F$, let $R(n, m)$ be the ring of invariant polynomials for the 
action of $\SL(n, \F) \times \SL(n, \F)$ on tuples of matrices -- $(A, C)\in\SL(n, 
\F) \times \SL(n, \F)$ sends $(B_1, \dots, 
B_m)\in M(n, \F)^{\oplus m}$ to $(AB_1C^{-1}, \dots, AB_mC^{-1})$. In this paper 
we call $R(n, m)$ the \emph{ring of matrix 
semi-invariants}. Let 
$\beta(R(n, m))$ be the smallest $D$ s.t. matrix semi-invariants of degree $\leq 
D$ generate $R(n, m)$. 
Guided by the Procesi-Razmyslov-Formanek approach of proving a strong degree bound 
for generating matrix invariants, we exhibit several interesting structural 
results for the 
ring of matrix semi-invariants $R(n, m)$ over fields of characteristic $0$. Using 
these results, we prove that $\beta(R(n, m))=\Omega(n^{3/2})$, and 
$\beta(R(2, m))\leq 4$. 
\end{abstract}

\section{Introduction}


Let $\F$ be a field. In this article, we study the ring of invariant polynomials 
for the action 
of $\SL(n, \F) \times \SL(n, \F)$ on tuples of matrices -- $(A, C)\in\SL(n, \F) 
\times \SL(n, \F)$ sends $(B_1, \dots, 
B_m)\in M(n, \F)^{\oplus m}$ to $(AB_1C^{-1}, \dots, AB_mC^{-1})$. Denoted 
by $R(n, m)$, we call this ring the \emph{ring of matrix semi-invariants}, as (1) 
it is closely related to the classical ring of matrix invariants \cite{Pro76} 
(see below for 
the definition, and \cite{Domokos00,ANK07} for the 
precise relationship between these two rings); and (2) it is the ring of 
semi-invariants of the representation of the 
$m$-Kronecker quiver with dimension vector $(n, n)$. Here, the $m$-Kronecker 
quiver is the quiver with two vertices $s$ and $t$, 
with $m$ arrows pointing from $s$ to $t$. When $m=2$, it is the classical 
Kronecker quiver. The reader is referred to \cite{DW00,SV01,DZ01} for a 
description of the semi-invariants for arbitrary quivers. 

Let $\beta(R(n, m))$ be the smallest 
integer $D$ s.t. matrix semi-invariants of degree $\leq D$ generate $R(n, m)$, and 
let $\sigma(R(n, m))$ be the smallest integer $D$ s.t. matrix semi-invariants of 
degree $\leq D$ define the nullcone of $R(n, m)$.  

Recently, the quantity $\sigma(R(n, m))$ has found several applications in 
computational complexity theory. In particular, if $\sigma(R(n, m))$ is polynomial 
in $n$ and $m$, then it follows that (1) division gates can be efficiently 
removed in non-commutative arithmetic 
circuits with division (\cite{HW14}); and (2) there exists a deterministic  
polynomial-time 
algorithm that decides whether the non-commutative rank of a square matrix of 
linear forms is full or not over the rational number field (\cite{IQS1,Gurvits}).
We refer the interested reader to the cited works for further explanation on these 
connections. At present, the best bound for $\sigma(R(n, m))$ is $n!/\lceil 
n/2\rceil !$ over large enough fields \cite{IQS1}.

One natural way to upper bound $\sigma(R(n, m))$ is of course to upper bound 
$\beta(R(n, m))$. Over algebraically closed fields of characteristic $0$, by 
Derksen's result \cite{derksen_bound}, $\beta(R(n, m))\leq \max(2, 3/8\cdot 
n^4\cdot \sigma(R(n, m))^2)$. Therefore, if $\sigma(R(n, m))$ is polynomial in $n$ 
then $\beta(R(n, m))$ is polynomial in $n$ as well.\footnote{Over fields of 
characteristic $0$, by a theorem of Weyl \cite{Wey97}, $\beta(R(n, m))\leq 
\beta(R(n, n^2))$. Therefore if $\beta(R(n, m))$ is polynomial in $n$ and $m$ then 
$\beta(R(n, m))$ is just polynomial in $n$. See \cite{Dom_description} for 
more on this. } Another compelling reason 
to examine $\beta(R(n, m))$ is because of the following strong upper bound on 
$\beta$ for the ring of matrix invariants over fields of characteristic $0$. 

Consider the action 
of $A\in\SL(n, \F)$ on $(B_1, \dots, B_m)\in M(n, \F)^{\oplus m}$ by 
sending it to $(AB_1A^{-1},$ $\dots, AB_mA^{-1})$. The invariant ring w.r.t. this 
action is denoted as $S(n, m)$, and elements in $S(n, m)$ are called \emph{matrix 
invariants}. The structure of $S(n, m)$ is 
well-understood over fields of characteristic $0$: the first fundamental theorem 
(FFT), the second fundamental theorem (SFT), and an $n^2$ upper 
bound for $\beta(S(n, m))$ have been 
established in 1970's by Procesi, Razmysolov, and Formanek  
\cite{Pro76,Raz74,formanek_gen}. Note that when applied to $S(n, m)$ over 
characteristic $0$, Derksen's bound yields $\beta(S(n, m))\leq \max(2, 3/8\cdot 
n^2\cdot 
\sigma^2)$ and $\sigma(S(n, m))=n^{O(n^2)}$, far from the $n^2$ 
bound as mentioned above. 

On the other hand, for $R(n, m)$, as far as we are aware, the best bound for 
$\beta(R(n, m))$ is 
$O(n^4\cdot (n!)^2)$ over algebraically closed fields of characteristic $0$, by 
combining the abovementioned results of \cite{derksen_bound} and \cite{IQS1}.

Our 
goal in this paper is to prove a better bound on $\beta(R(n, m))$ by following the 
approach of Procesi, Razmyslov, and Formanek. Therefore, in the following we 
restrict ourselves to fields of characteristic $0$. While we do not achieve 
this, we describe several structural results for $R(n, m)$, 
including the second fundamental theorem (Proposition~\ref{prop:sft}). Though some 
of these 
structural results should be known to experts, we could not find 
them in the literature, so we provide full proof details. Furthermore, these 
results allow us to prove that $\beta(R(n, m))$ has a lower bound 
$\Omega(n^{3/2})$ (Proposition~\ref{prop:lower}).

One technical result that we believe is new, is 
Proposition~\ref{prop:simple_span}. Roughly speaking, there exists a linear basis 
of matrix semi-invariants, such that each polynomial in this basis can be 
associated with a bipartite graph. An upper bound of $D$ on $\beta(R(n, m))$ would 
follow, if we can 
prove that when the degree is $>D$, modulo the linear relations (as described in 
the second 
fundamental theorem), every polynomial can be written as a linear combination of 
those basis elements whose associated graphs are disconnected 
(Proposition~\ref{prop:deg}). 
Proposition~\ref{prop:simple_span} then states that when $D>n^2$, every matrix 
semi-invariant of degree $D$ 
can be written as a linear combination of those basis elements whose associated 
graphs are disconnected \emph{or non-simple} (e.g. with at least one multiple 
edges) modulo the linear relations. 

As an immediate consequence of Proposition~\ref{prop:simple_span}, we prove that 
$\beta(R(2, m))\leq 4$  over fields of characteristic 
$0$ in Theorem~\ref{thm:2times2}. While this bound is known from 
Domokos' explicit generating set for $\beta(R(2, m))$ \cite{Domokos00_2}, we think 
this suggests 
the validity of the approach of Razmyslov, Procesi and Formanek when applied to 
matrix semi-invariants. Furthermore, since in this approach we do not 
exhibit invariants explicitly, we believe this approach will generalize to 
larger $n$. 

\paragraph{More previous works.} Here we collect a few more previous 
works on matrix semi-invariants. To start with, since as mentioned, matrix 
semi-invariants are just semi-invariants of the representations of the $m$-Kronecker quiver, 
results on semi-invariants of quivers apply to matrix semi-invariants, e.g. the 
first fundamental theorem \cite{DW00,SV01,DZ01}.
Let us give one description from \cite{DZ01}: suppose $R(n, 
m)\subseteq \F[x_{i,j}^{(k)}]$ where $i, j\in[n]$, $k\in [m]$, and $x_{i,j}^{(k)}$ 
are independent variables. Let $X_k=(x_{i,j}^{(k)})$. Then for $A_1, 
\dots, A_m \in M(d, \F)$, $\det(A_1\otimes X_1+\dots+A_m\otimes X_m)$ is a matrix 
semi-invariant, and every matrix semi-invariant is a linear combination of such 
polynomials. Therefore, $(B_1, \dots, B_m)$ is in the nullcone, if and only if for 
all $d\in\Z^+$ and all $(A_1, \dots, A_m)\in M(d, \F)^{\oplus m}$, $A_1\otimes B_1
+\dots+A_m\otimes B_m$ is singular. 

A description of the nullcone of $R(n, m)$ can be found in \cite{BD06,ANK07}. 
For certain small $m$ or $n$, explicit generating sets of 
$R(n, m)$ 
have been computed in e.g. \cite{Domokos00,Domokos00_2,DD12}. In these cases, 
elements of degree $\leq n^2$ generate the 
ring.\footnote{We thank M. Domokos for pointing out this fact to us.}

Several results for matrix invariants over fields of positive characteristics are 
known: FFT was established by Donkin in \cite{Donkin92,Donkin93}, an 
$n^3$ upper bound for $\sigma$ can be 
derived from \cite[Proposition 9]{CIW}, and 
Domokos in 
\cite{Domokos02a,Domokos02b} proved an upper bound $O(n^7 m^n)$ on $\beta$.

\paragraph{Organization.} In Section~\ref{sec:overview} we briefly review the 
Procesi-Razmyslov-Formanek approach, and give an overview of our results. Then we 
describe 
the second fundamental theorem (Section~\ref{sec:sft}), the $\F \sg_{dn}$ bimodule 
structure (Section~\ref{sec:bimodule}), and the $\sg_{dn}$ diagonal action 
(Section~\ref{subsec:diag}). Finally, in Section~\ref{sec:deg}, we use these 
structural results to prove a lower bound on $\beta(R(n, m))$, and that 
$\beta(R(2, m))\leq 4$. 


\section{An overview of the structural results}\label{sec:overview}

\paragraph{An outline of the Procesi-Razmyslov-Formanek result.} To motivate the 
results to be presented, we first review the $n^2$ bound 
for $S(n, m)$ over characteristic-$0$ fields by Razmyslov \cite{Raz74} and Procesi 
\cite{Pro76}, and 
further elaborated by Formanek \cite{formanek_gen}. Recall 
that $S(n, m)$ is the invariant ring of $A\in 
\SL(n, \Q)$ on $(B_1, \dots, B_m)\in M(n, \Q)^{\oplus m}$ by sending it to 
$(AB_1A^{-1},$ 
\dots, $AB_mA^{-1})$. Our exposition follows the one by Formanek 
\cite{formanek_gen}, and requires certain basic facts about the group 
algebra of the symmetric group as described therein.


Let $X_i$, $i\in[m]$ be an $n\times n$ matrix of indeterminants. 
Then $\Tr(X_{i_1}\cdot X_{i_2}\cdot \dots \cdot X_{i_k})$, $k\in \Z^{+}$, $i_j\in 
[m]$ generate all matrix invariants. The $n^2$ upper bound implies that 
invariants of this form with $k\leq n^2$ generate the ring of matrix invariants 
already. 

The proof for this upper bound starts with identifying multilinear matrix 
invariants in the form $\Tr(X_{i_1}\cdot X_{i_2}\cdot \dots \cdot X_{i_k})$, 
$i_j\in[m]$ pairwise distinct, with the permutations of a single cycle $(i_1, i_2, 
\dots, i_k)$. Such identification gives a correspondence of multilinear invariants 
of degree $d$ with the group algebra $\F\sg_d$. 

The second step is to use the second fundamental theorem of matrix invariants, 
which suggests that the linear relations are spanned by the two-sided ideals 
indexed by partitions of length $> n$. At this 
point, it is clear that $b$ is a degree bound, if and only if, for any $d>b$, the 
reducible 
elements (those permutations whose cycle types are not of a single cycle of length 
$d$) together with the linear relations span the whole space. Then, by using the 
standard bilinear form for $\C\sg_d$ (setting the permutations as an orthonormal 
basis), this is equivalent to showing that 
the space $J$ spanned by the irreducible elements (those permutations whose cycle 
types are a single cycle of length $d$), and the space $I_{\leq n}$ spanned by the 
two-sided 
ideals indexed by partitions of length $\leq n$, intersect trivially. 

The third step is to observe that the linear map $f$ defined by $\sigma \to 
\sgn(\sigma)\sigma$ where $\sigma\in \sg_d$, sends the two-sided ideal 
corresponding to $\lambda$, to the two-sided ideal corresponding to 
$\conjpart{\lambda}$, the conjugate of 
$\lambda$. At the same time, $f$ preserves every $1$-dimensional subspace in $J$, 
as $\sgn(\sigma)$ for any permutation $\sigma\in J$ is the same. Thus when $d > 
n^2$, $J$ and $I_{\leq n}$ have to intersect trivially: by contradiction, suppose 
a nonzero $v\in J \cap I_{\leq n}$. Then $f(v)=\pm v$, thus 
$v\in I_{\leq 
n} \cap f(I_{\leq n})$. However, $f(I_{\leq n})$ is spanned by the two-sided 
ideals of width $\leq n$. When $d > n^2$, there does not exist a partition with 
both length and width $\leq n$. The $n^2$ degree bound then follows.

\paragraph{Overview of the structural results of $R(n, m)$.}
To carry on the above strategy for $S(n, m)$, a first step is to give the 
multilinear invariants in $R(n, m)$ a combinatorial description, which we take 
from \cite{ANK07}. Since in this case we have multilinear invariants only when  
$n$ divides $m$, we focus on $R(n, dn)$ in the following. Briefly speaking, 
we can identify a natural spanning set of multilinear invariants with $n$-regular 
bipartite graphs with $d$ left (resp. right) vertices. 

Then it is necessary to obtain the second fundamental theorem (SFT) for $R(n, 
dn)$. While certainly known to experts, we could not find an explicit statement in 
the literature, so we prove it in Proposition \ref{prop:sft}. Then we 
need to describe the $\F\sg_{dn}$ bimodule structure of $R(n, 
dn)$ (Fact~\ref{fact:bimodule}).
This is 
possible since $R(n, dn)$ can be viewed as a tensor product of two column tabloid 
modules \cite[Chap. 7.4]{Fulton}. So $b$ is a degree bound, if 
whenever $dn>b$, with the help of relations, those multilinear invariants 
indexed by 
connected bipartite graphs can be written as a sum of those ones indexed by 
disconnected bipartite graphs (Proposition~\ref{prop:deg}). 
Furthermore, a natural bilinear form on 
$R(n, dn)$ can be defined by setting the standard elements in $R(n, dn)$ as 
forming an orthonormal basis (Fact~\ref{fact:bimodule} (3)). This bilinear form 
makes sense 
w.r.t. the bimodule 
structure as well due to James' submodule theorem (Fact~\ref{fact:submodule}). Up 
to 
this 
point, we successfully parallel everything in the $S(n, m)$ setting. 

What is missing is an analogue of the linear map $f$ - this prevents us from completing this 
strategy. Nevertheless, what we've develop allows us to prove that for $b$ to be 
a degree bound, then $b$ must be $\Omega(n^{3/2})$ (Proposition~\ref{prop:lower}). 

To make progress, we exploit in depth the diagonal action of $\sg_{dn}$ on $R(n, 
dn)$. Fortunately, the orbits under this action can be identified with 
row tabloid modules \cite[Chap. 7.2]{Fulton}. From 
this structure we obtain new forms of relations (Equation \ref{eq:alt}). We then prove 
that when $b>dn$, using these relations, each multilinear invariant indexed by a 
connected 
and simple (e.g. with no multiple edges) bipartite graph, can be written as a sum 
of those 
indexed by non-simple graphs (Proposition~\ref{prop:simple_span}). This provides a 
non-trivial 
reduction result in 
the spirit of the degree bound statement. 

We believe that the results we prove in this paper, will be 
useful to finally get a good degree bound for $R(n, m)$, over fields of 
characteristic $0$. 

\section{The second fundamental theorem}\label{sec:sft}

\paragraph{Reduction to the multilinear case.} Over characteristic $0$ fields, the 
well-known two procedures, polarization and restitution, reduce many questions for 
general invariants to multilinear invariants.

For completeness, let us demonstrate this in the case of $R(n, m)\subseteq 
\F[x^{(k)}_{i,j}]_{k\in[m], i,j\in[n]}$. Let $X_k=(x^{(k)}_{i,j})_{i,j\in[n]}$. 
Recall that each $f\in 
R(n, m)$ has degree divisible by $n$. Suppose $\ell(X_1, \dots, X_m)=\sum_i a_i 
f_i$ forms a relation, where $f_i\in R(n, m)$. W.l.o.g. we can assume the $f_i$'s 
are multi-homogeneous, that is for every fixed $j\in[m]$, the degrees of $f_i$'s 
w.r.t. $X_j$ are the same for every $i$. Otherwise, we can divide $\ell$ into 
the multi-homogeneous 
components, and each component will again be a relation. Suppose $\deg(X_k)$ 
in $\ell$ is $d_k$, then $\sum_k d_k=dn$ for some $d\in\Z^+$. Then for each 
$k\in[m]$, we introduce $d_k$ copies of $X_k$ as $X_{(k, 1)}, \dots, X_{(k, 
d_k)}$, as well as $d_k$ variables $y_{k, 1}, \dots, y_{k, d_k}$. Now consider 
the 
coefficient of $\prod_{i,j}y_{i,j}$ in $\ell(X_{(1, 1)}y_{1,1}+\dots+X_{(1, d_1)} 
y_{1, d_1}, \dots, X_{(m, 1)}y_{m,1}+\dots+X_{(m,d_m)}y_{m,d_m})$, and let it be 
$\ell'$ -- the polarization of $\ell$. It can be seen that $\ell'$ is a multilinear 
invariant in $dn$ matrices, and $\ell'$ is a relation as well. Now we modify 
$\ell'$ by substituting $X_{(i, j)}$ by $X_i$, and let the result be $\ell''$ -- the 
substitution of $\ell'$. Over characteristic $0$ fields, we see that 
$\ell''=\prod_{i}(d_i!) \ell$. As demonstrated above, since every relation can be 
obtained by restituting some multilinear relation, it is enough to understand 
multilinear relations.

%
\paragraph{Multilinear invariants of $R(n, m)$.} 
As invariants in $R(n, m)$ are of 
degree divisible by $n$, $R(n, m)$ has a multilinear invariant if and only if $n$ 
divides $m$. In the following we consider $R(n, m)$ for $m=dn$ where $d\in \Z^+$.
We use the descriptions of invariants in $R(n, m)$ given in \cite{ANK07}. 
We reformulate their results here. 

We define a set of symbols $\poset=\po{i_1, 
\dots, i_n}$, $i_j\in[dn]$ with the anti-symmetric property: 
$$\po{i_1, \dots, i_a, \dots, i_b, \dots, i_n}=-\po{i_1, \dots, i_b, \dots, i_a, 
\dots, i_n}.$$
Note that as $\F$ is of characteristic $0$, the anti-symmetric property implies 
that if there exist $i_j=i_k$ for $j\neq k$ then $\po{i_1, \dots, i_n}=0$.

Then let $\posetring$ be $\F\langle \poset\rangle$, the noncommutative polynomial 
ring with variables from $\poset$. A degree-$d$ monomial in $\posetring$ is of 
the form 
$$
\po{i_{1,1}, \dots, i_{1,n}}\circ \dots \circ \po{i_{d,1},\dots,i_{d,n}},
$$ 
where $\circ$ denotes the noncommutative product. 
We often record it as an $d\times n$ tableau
\begin{equation}\label{eq:tab}
\begin{vmatrix}
i_{1,1} & i_{1,2} & \cdots & i_{1,n} \\
  i_{2,1} & i_{2,2} & \cdots & i_{2,n} \\
  \vdots  & \vdots  & \ddots & \vdots  \\
  i_{d,1} & i_{d,2} & \cdots & i_{d,n} 
\end{vmatrix}.
\end{equation}

A degree-$d$ monomial $\po{i_{1,1}, \dots, i_{1,n}}\circ \dots \circ 
\po{i_{d,1},\dots,i_{d,n}}$ is repetition-free, if $\{i_{j,k}, j\in[d], 
k\in[n]\}=[nd]$. Let 
$\posetring(d)$ be the $\F$-vector space of spanned by degree-$d$ repetition-free 
monomials in $\posetring$. 

Likewise we define another set of symbols $\pair{\poset}=\po{\pair{i_1}, \dots, 
\pair{i_n}}$, $i_j\in[nd]$ satisfying also the anti-symmetric property, and define
$\pair{\posetring}$, and $\pair{\posetring}(d)$ as before. The vector space 
$\posetring(d)\otimes \pair{\posetring}(d)$ then has a basis $S \otimes \pair{T}$, 
$S\in \posetring(d)$ and $\pair{T}\in\pair{\posetring}(d)$. 


Now consider the vector space $(V\otimes W)^{\oplus dn}$ where $V\cong W\cong 
\F^n$. This is spanned by vectors of the form $(v_1\otimes 
w_1) \oplus \dots \oplus (v_{dn}\otimes w_{dn})$, $v_i\in V$, $w_j\in W$. Given 
$i_1,i_2,\ldots,i_n \in[dn]$, we 
define a function $\por{i_1, \dots, i_n}$, on $(v_1\otimes w_1) \oplus \dots \oplus (v_{dn}\otimes 
w_{dn})$, as follows. $\por{i_1, 
\dots, i_n}$ sends $(v_1\otimes 
w_1) \oplus \dots \oplus (v_{dn}\otimes w_{dn})$ to the determinant of the 
$n\times n$ matrix $[v_{i_1}, \dots, v_{i_n}]$. Likewise, define $\por{\pair{j_1}, 
\dots, \pair{j_n}}$, $j_k\in[dn]$ sending $(v_1\otimes w_1) \oplus \dots \oplus 
(v_{dn}\otimes w_{dn})$ to the 
determinant of the matrix $[w_{j_1}, \dots, w_{j_n}]$. $\por{i_1, \dots, i_n}$ is 
known as a \emph{Pl\"ucker coordinate}. We then consider the product function
$$\por{i_{1,1}, \dots, i_{1,n}}\cdot \ldots \cdot \por{i_{d,1},\dots,i_{d,n}} 
\cdot 
\por{\pair{j_{1,1}}, \dots, \pair{j_{1,n}}}\cdot \ldots \cdot 
\por{\pair{j_{d,1}},\dots,\pair{j_{d,n}}}$$
where $\{i_{k, \ell}, k\in[d], \ell\in[n]\}=\{j_{k, \ell}, k\in[d], 
\ell\in[n]\}=[nd]$, and extend this function to all of $(V\otimes W)^{\oplus dn}$ by linearity. 
This gives rise to a multilinear function on $(V\otimes W)^{\oplus dn}$.

We then define a linear map $\phi$ from $\posetring(d)\otimes 
\pair{\posetring}(d)$ to multilinear functions on $(V\otimes W)^{\oplus 
dn}$, by sending 
\begin{multline}\label{eq:phi}
(\po{i_{1,1}, \dots, i_{1,n}}\circ \ldots \circ \po{i_{d,1},\dots,i_{d,n}}) 
\otimes 
(\po{\pair{j_{1,1}}, \dots, \pair{j_{1,n}}}\circ \ldots \circ 
\po{\pair{j_{d,1}},\dots,\pair{j_{d,n}}}) \\
\text{to} \quad 
\por{i_{1,1}, \dots, i_{1,n}}\cdot \ldots \cdot \por{i_{d,1},\dots,i_{d,n}} 
\cdot 
\por{\pair{j_{1,1}}, \dots, \pair{j_{1,n}}}\cdot \ldots \cdot 
\por{\pair{j_{d,1}},\dots,\pair{j_{d,n}}}
\end{multline}
and extend by linearity. It is understood that if we write $\phi$ to apply to 
monomials in $\posetring(d)$ or $\pair{\posetring}(d)$, we replace 
$\po{\cdot}$ to $\por{\cdot}$. It is not difficult to observe (see e.g. 
\cite{ANK07}) that the image of 
$\phi$ are precisely the multilinear functions on $(V\otimes W)^{\oplus 
dn}$ which are invariant under the natural action $\SL(n, \F)\times \SL(n, \F)$. 

The formalism above is necessary for the formulation of the second fundamental 
theorem, which amounts to describe the kernel of $\phi$.

\paragraph{Second fundamental theorem for multilinear invariants.} 
In this part we describe the kernel of $\phi$. It is clear that when $d=1$, 
$\ker(\phi)$ is $0$. In the following we assume $d\geq 2$. 

Recall the Pl\"ucker relations for Pl\"ucker coordinates: for $i_1, \dots, i_n, 
j_1, \dots, j_n\in[nd]$ and $k\in[n]$, we have
\begin{eqnarray*}
& & \por{i_1, \dots, i_n} \cdot \por{j_1,\dots, j_n} \\
& = & \sum_{1\leq s_1 < s_2 < \dots< s_k\leq n} \por{i_1, \dots, j_1, \dots, j_k, 
\dots, i_n} \cdot \por{i_{s_1}, \dots, i_{s_k}, j_{k+1}, \dots, j_n},
\end{eqnarray*}
where $j_1, \dots, j_k$ are in positions $s_1, \dots, s_k$. 
For a monomial $T$ in $\posetring(d)$ recorded as in Equation~\ref{eq:tab},
following Fulton \cite[pp. 97]{Fulton}\footnote{With a slight change from column 
as in \cite{Fulton} to row here.}, we use $\pi_{j, k}(T)$ to denote the 
vector in 
$\posetring(d)$ of the form $\sum T'$, where $T'$ runs over the tableaux obtained 
by switching the first $k$ elements in the $j+1$th row, with $k$ elements in the 
$j$th row of $T$. Let $\kpo(d)$ be the span of all $T-\pi_{j,k}(T)$, where 
$T\in\posetring(d)$, $j\in[n-1]$, and $k\in[n]$. In particular, note that when 
$k=n$, $\pi_{j, n}(T)$ is just switching the $j$th and $j+1$th row. So though 
monomials in $\posetring(d)$ are non-commutative, modulo $\kpo(d)$ they become 
commutative, which is consistent with the image of $\phi$.

A monomial $T$ in $\posetring(d)$ is \emph{standard}, if its tableau is (strictly) 
increasing in each row and each column. Following a procedure called the 
straightening, it is well-known that the standard 
monomials form a basis for the quotient space $\posetring(d)/\kpo(d)$. 

Similarly we have $\pair{\kpo}(d)$ in $\pair{\posetring}(d)$. By Pl\"ucker 
relations, we know that $\langle \kpo(d)\otimes \pair{\posetring}(d) \cup 
\posetring(d)\otimes \pair{\kpo}(d)\rangle$ lies in $\ker(\phi)$. (Recall that 
$\langle \cdot \rangle$ denotes the linear span.) We show that these two sets are 
equal.

\begin{proposition}\label{prop:sft}
Let notations be as above. We have $\langle \kpo(d)\otimes \pair{\posetring}(d) 
\cup \posetring(d)\otimes \pair{\kpo}(d)\rangle=\ker(\phi)$.
\end{proposition}
\begin{proof}
It remains to prove that $\ker(\phi)\subseteq \langle \kpo(d)\otimes 
\pair{\posetring}(d) \cup \posetring(d)\otimes \pair{\kpo}(d)\rangle$. 
Let $\ell$ be in $\ker(\phi)$. Suppose $\ell$ is 
$$c_1 
S_1\otimes \pair{T_1}+c_2 S_2\otimes \pair{T_2}+\dots+c_k S_k\otimes 
\pair{T_k}$$ 
where $c_i$'s are in $\F$, and $S_i$'s (resp. $\pair{T_i}$'s) are monomials in 
$\posetring(d)$ (resp. $\pair{\posetring}(d)$). Note that $S_i$'s (resp. 
$\pair{T_i}$'s) are not necessarily distinct. We first arrange according to 
$S_i$'s, and write $\ell$ as 
$$
S_1\otimes h_1+ S_2\otimes h_2+\dots + S_{k'}\otimes h_{k'}
$$ 
where $h_i\in \pair{\posetring}(d)$, and $S_1$, \dots, $S_{k'}$ are distinct. 
Modulo the space $\kpo(d) \otimes \pair{\posetring}(d)$, we 
 express each $S_i$ as a linear combination of those ones indexed by standard 
tableaux. Then by re-grouping according to the standard tableaux, $\ell$ is 
expressed as 
$$
S_1'\otimes h_1'+ S_2'\otimes h_2'+\dots + S_p'\otimes h_p'
$$
where $h_i'\in \pair{\posetring}(d)$, and $S_1'$, \dots, $S_p'$ are standard and 
distinct. Now we claim that every 
$S_i'\otimes h_i'$ must be in $\posetring(d) \otimes \pair{\kpo}(d)$. W.l.o.g. 
assume $S_1'\otimes h_1'$ is not, that is $h_1'$ is not in $\pair{\kpo}(d)$. Then 
$\phi(h_1')$ is not a zero function on $W^{\oplus nd}$. So there exists some 
$\vec{w}=w_1\oplus w_2\oplus \dots \oplus 
w_{nd}$ s.t. $\phi(h_1')(\vec{w})\neq 0$. Now we restrict $\phi(\ell)$ to the 
set  $(V\otimes w_1)\oplus \dots\oplus (V\otimes w_{nd})$ which yields 
$$
\phi(S_1')\cdot \phi(h_1')(\vec{w})+ \phi(S_2')\cdot \phi(h_2')(\vec{w})+\dots + 
\phi(S_j')\cdot \phi(h_j')(\vec{w}).
$$
Since $S_i'$'s are standard and distinct, 
$\phi(S_i')$'s are linearly independent as functions on $V^{\oplus dn}$. As 
$\phi(h_1')(\vec{w})\neq 0$, the
restriction of $\ell$ on $(V\otimes w_1)\oplus \dots\oplus (V\otimes w_{nd})$, 
as a 
function on $V^{\oplus d}$, is nonzero.
This contradicts $\ell$ being a zero function on $(V\otimes W)^{\oplus 
dn}$. So every $S_i'\otimes h_i'$ is in $\posetring \otimes \pair{\kpo}$. From 
this we conclude that $\ell\in \langle \kpo(d)\otimes \pair{\posetring}(d) \cup 
\posetring(d)\otimes \pair{\kpo}(d)\rangle$.
\end{proof}

\section{The $\F\sg_{dn}$ bimodule structure}\label{sec:bimodule}

There exists a natural action of $\sg_{dn}\times \sg_{dn}$ on 
$\posetring(d)\otimes \pair{\posetring}(d)$. Suppose $S$ and $\pair{T}$ are 
monomials in $\posetring$ and $\pair{\posetring}$ respectively. Then $(\sigma, 
\tau)\in\sg_{dn}\times\sg_{dn}$ sends $S\otimes \pair{T}$ to $S^{\sigma}\otimes 
\pair{T}^{\tau}$, where $\sigma$ and $\tau$ permute the entries of $S$ and 
$\pair{T}$. This endows $\posetring(d)\otimes \pair{\posetring}(d)$ an 
$\F\sg_{nd}$ bimodule structure. This action descends to the relations describing the second 
fundamental theorem: if $\ell$ is a relation (in $\ker(\phi)$), then 
$\ell^{(\sigma, \tau)}$ is also a relation. So $\ker(\phi)$ is an $\F\sg_{dn}$ 
sub-bimodule, whose structure will be described in the following 
Fact~\ref{fact:bimodule}.

Such an action on $\sg_{dn}$ on $\posetring(d)$ is well-understood: this 
is the column tabloid module as discussed in \cite[Chap. 
7.4]{Fulton}.\footnote{Note that the notation in \cite{Fulton} is a bit different 
from ours; there the column tableux satisfy anti-symmetric property along the 
columns. So one needs to flip the tableaux here to match the results there.} The 
column 
tabloid module is a dual of the more well-known row tabloid module as shown in 
\cite[Chap. 7.2]{Fulton}. We collect some basic facts about the column tabloid 
module adapted to our setting. Recall that for a nonnegative integer $s$, a 
partition of size $s$ is a non-increasing sequence of nonnegative integers 
$\lambda=(\lambda_1, \dots, \lambda_\ell)$, with 
$\sum_{i=1}^\ell\lambda_i=s$. This is denoted as 
$\lambda \vdash s$. We identify $(\lambda_1, \dots, \lambda_\ell)$ with 
$(\lambda_1, \dots, \lambda_\ell, 0, \dots, 0)$; namely any trailing zeros are 
assumed implicitly if required. The conjugate of $\lambda$ denoted by 
$\conjpart{\lambda}$ is a partition of $s$ with 
$\conjpart{\lambda}_i=|\{j\in[\ell] \text{ s.t. }\lambda_j\geq i\}|$. The height 
of $\lambda$ is $\max(i\in[\ell]\mid \lambda_i\neq 0)$, and the width is 
$\lambda_1$. Partitions are usually 
represented using Young diagrams: that is a concatenation of rows of boxes 
arranged to be left aligned, with the $i$th row having $\lambda_i$ boxes. For two 
partitions 
$\nu=(\nu_1, \dots, \nu_k)$ and $\lambda=(\lambda_1, \dots, \lambda_\ell)$, $\nu$ 
dominates $\lambda$ if for any $j\in\Z^+$, $\sum_{i=1}^j \nu_i\geq \sum_{i=1}^j 
\lambda_i$. We use $d^n$ to denote the partition $(d, \dots, d)$ of height $n$. 
For a partition $\lambda$ of $[dn]$, $\specht(\lambda)$ denotes the Specht module 
corresponding to $\lambda$ over $\F$. 
\begin{fact}[{\cite[Chap. 7.4]{Fulton}}]\label{fact:column_tabloid}
\begin{enumerate}
\item As an $\F\sg_{dn}$ module, $\posetring(d)$ 
decomposes as $$\specht(d^n)\oplus 
\big(\bigoplus_{\lambda\vdash dn}\specht(\lambda)^{\oplus c_{\lambda}}\big),$$
where $\lambda$ runs over partitions of $dn$ strictly dominated by $d^n$, and 
$c_\lambda=K_{\conjpart{\lambda}, n^d}$, the Kostka number w.r.t. 
$\conjpart{\lambda}, n^d$.
\item Given a monomial $S\in\posetring(d)$ represented by a tableau, the 
column subgroup of $S$, $\colgp(S)$ is the subgroup of $\sg_{dn}$ that 
preserves the columns of $S$. The (unsigned) column symmetrizer 
$\colsym(S)=\sum_{\pi\in \colgp(S)}\pi$. Then $\ipo(d):=\langle \colsym(S)\cdot S 
\mid S\in\posetring(d) \text{ are monomials}\rangle\cong \specht(d^n)$.
\item $\kpo(d)$ is a submodule of $\posetring(d)$ and is isomorphic to 
$\bigoplus_{\lambda\vdash dn}\specht(\lambda)^{\oplus c_{\lambda}}$.
\end{enumerate}
\end{fact}

We also observe the following adaptation of James' submodule theorem to column 
tabloid modules.

\begin{fact}[{See e.g. \cite[Theorem 2.4.4]{Sagan}}]\label{fact:submodule}
Let $\beta'$ be the bilinear form on $\posetring(d)$ by setting the monomials as 
an 
orthonormal basis. Then $\kpo(d)$ and $\ipo(d)$ are orthogonal complement of each 
other under $\beta'$.
\end{fact}

Thus, as an $\F \sg_{dn}$ bimodule, $\posetring(d)\otimes \pair{\posetring}(d)$ is 
the tensor product of two column tabloid modules. Its structure is then easily 
deduced from Fact~\ref{fact:column_tabloid} and~\ref{fact:submodule}.

\begin{fact}\label{fact:bimodule}
\begin{enumerate}
\item As an $\F\sg_{dn}$ bimodule, $\posetring(d)\otimes\pair{\posetring}(d)$ 
decomposes as 
$$
(\specht(d^n)\otimes\specht(d^n)) \oplus \big(\bigoplus_{\lambda, \nu\vdash 
dn}(\specht(\lambda)\otimes\specht(\nu))^{\oplus c_\lambda\cdot c_\nu}\big),
$$
where $\lambda$ and $\nu$ are dominated by $d^n$, and at least one of $\lambda$ 
and $\nu$ is strictly dominated by $d^n$. $c_\lambda$ 
and $c_\nu$ are Kostka numbers as in Fact~\ref{fact:column_tabloid}.
\item $\ker(\phi)=\langle \kpo(d)\otimes\pair{\posetring}(d)\cup 
\posetring(d)\otimes \pair{\kpo}(d)\rangle$ is an $\F\sg_{dn}$ sub-bimodule, 
isomorphic to the summand $\bigoplus_{\lambda, \nu\vdash 
dn}(\specht(\lambda)\otimes\specht(\nu))^{\oplus c_\lambda\cdot c_\nu}$, as 
above.
\item Let $\beta$ be the bilinear form on 
$\posetring(d)\otimes\pair{\posetring}(d)$ 
by setting $S\otimes \pair{T}$, where $S$ and $\pair{T}$ are monomials in 
$\posetring(d)$ and $\pair{\posetring}(d)$, respectively. Then 
$\ipo(d)\otimes\pair{\ipo}(d)$ and $\langle \posetring(d)\otimes \pair{\kpo}(d) 
\cup \kpo(d)\otimes\pair{\posetring}(d)\rangle$ are orthogonal complement of each 
other under $\beta$.
\end{enumerate}
\end{fact}

\section{The diagonal action of $\sg_{dn}$}\label{subsec:diag}

In this subsection we consider the diagonal action of 
$\sg_{dn}$ on $\posetring(d)\otimes\pair{\posetring}(d)$. That is, 
$\sigma\in\sg_{dn}$ acts on 
$\posetring(d)\otimes\pair{\posetring}(d)$ as $(\sigma, \sigma)\in 
\sg_{dn}\times\sg_{dn}$. 

To understand this structure we introduce a combinatorial structure called the 
\emph{correlated tableaux}. A correlated tableau $C$ of size $d\times d$ is a 
$d\times d$ matrix whose entries are subsets of $[dn]$, satisfying: (1) $C(i,j)$, 
$i, j\in[d]$ form a partition of $[dn]$; (2) $\forall i\in[d]$, $|\cup_{j\in[d]} 
C(i, j)|=n$; (3) $\forall j\in[d]$, $|\cup_{i\in[d]} C(i,j)|=n$.

Correlated tableaux are in 1-to-1 correspondence with $S\otimes\pair{T}$ where $S$ 
and $\pair{T}$ are monomials from $\posetring(d)$ and 
$\pair{\posetring}(d)$, respectively. Given $S\otimes\pair{T}$, we can form a 
correlated tableau $C$ by setting $C(i, j)=\{k\in[nd]\mid k\in S_i \text{ and } 
\pair{k}\in T_j\}$. Given a correlated tableau $C$, by reading along the rows from 
top to bottom we get a monomial $S$, and by reading along the columns from left to 
right (and adding $\pair{\ \cdot\ }$) we get a monomial $\pair{T}$.

Let us set up the following convention about arranging the rows of $S$ (resp. 
$\pair{T}$) as 
in $S\otimes\pair{T}$, where $S$ and $\pair{T}$ are understood as tableaux as 
shown in Equation~\ref{eq:tab}. In the literature, it 
is more common to arrange each row of $S$ to be increasing, 
namely $i_{j,1}<i_{j,2}<\dots <i_{j,n}$ for every $j$. But when $S$ appears in 
$S\otimes\pair{T}$ we shall order each row while taking into consideration of 
$\pair{T}$ as well. That is, first form the correlated tableau $C$ based on 
$S\otimes \pair{T}$, and then read each row of $C$ from left to right. If there 
are more than $1$ elements in an entry, then use the increasing order. The rows of 
$\pair{T}$ are also ordered in a similar way. 

As an example, consider 
$$
S=\begin{vmatrix}
1 & 2 & 4 \\
  3 & 5 & 6 \\
  7 & 8 & 9
\end{vmatrix}, \pair{T}=
\begin{vmatrix}
\pair{1} & \pair{2} & \pair{7} \\
  \pair{3} & \pair{4} & \pair{5} \\
  \pair{6} & \pair{8} & \pair{9}
\end{vmatrix}.
$$
The correlated tableaux then is $C=$
\begin{tabular}{|>{$}c<{$}|>{$}c<{$}|>{$}c<{$}|}
\hline \{1,2\} & \{4\} & \emptyset \\ 
\hline \emptyset & \{3,5\} & \{6\} \\ 
\hline \{7\} & \emptyset & \{8,9\} \\ 
\hline 
\end{tabular}.
So we obtain by reading $C$
$$
S'=\begin{vmatrix}
1 & 2 & 4 \\
  3 & 5 & 6 \\
  7 & 8 & 9
\end{vmatrix}, \pair{T'}=
\begin{vmatrix}
\pair{1} & \pair{2} & \pair{7} \\
  \pair{4} & \pair{3} & \pair{5} \\
  \pair{6} & \pair{8} & \pair{9}
\end{vmatrix}.
$$
That is $\pair{T'}=-\pair{T}$.

Having defined correlated tableaux, we define the $(d, n)$-\emph{correlated 
diagram} as a $d\times d$ matrix $D$ over $\N$ with row and column sums 
being $n$. For a correlated diagram $D$, we form its \emph{diagonal 
partition}
$\diagpart(D)$, by arranging the entries of $D$ to be non-increasing to get a 
partition of $dn$. Each correlated tableau $C$ gives rise to a correlated diagram 
$D_C$ by setting $D_C(i,j)=|C(i,j)|$; we all call $D_C$ the shape of $C$. Using 
the example as above, $D_C=$
\begin{tabular}{|>{$}c<{$}|>{$}c<{$}|>{$}c<{$}|}
\hline 2 & 1 & 0 \\ 
\hline 0 & 2 & 1 \\ 
\hline 1 & 0 & 2 \\ 
\hline 
\end{tabular}. $\diagpart(D_C)$ then is $(2, 2, 2, 1, 1, 1)$.

Fix a correlated diagram $D$, and let $\ctab(D)$ be the subspace of 
$\posetring(d)$, spanned by all correlated tableaux of shape $D$. $\ctab(D)$ is 
clearly a submodule of $\posetring(d)$ under the diagonal action of $\sg_{dn}$. 
Let us define the \emph{row tabloid module} of $\F\sg_{s}$, $s\in\Z^+$. For a 
partition $\lambda\vdash s$, row tabloids of shape $\lambda$ are fillings of the 
Young diagram of shape $\lambda$ by $[s]$ without repetitions, with the following 
equivalence relation: two row tabloids $P$ and $Q$ are equivalent if the 
corresponding rows have the same entries. Row tabloids admit an $\F\sg_{s}$ action 
by permuting the entries. The linear span of all row tabloids of shape $\lambda$ 
is then an $\F\sg_{s}$ module and denoted as $\rowtab(\lambda)$. The decomposition 
of the row tabloid 
module is well-known and quite similar to that of the column tabloid module as 
shown in Fact~\ref{fact:column_tabloid}; see \cite[Chap. 7.2]{Fulton}.

\begin{lemma}
As an $\F\sg_{dn}$ module, $\ctab(D)$ is isomorphic to 
the row tabloid module $\rowtab(\diagpart(D))$. 
\end{lemma}
\begin{proof}
To set up a linear map between $\ctab(D)$ and $\rowtab(\diagpart(D))$, arrange the 
entries of $C\in\ctab(D)$ following an arbitrary but fixed 
order as long as this order maintains the sizes of the entries to be 
non-increasing. This gives a row tabloid in $\rowtab(\diagpart(D))$, and then 
extend by linearity. To see that the actions of $\sg_{dn}$ are compatible, note 
that by our convention of sending correlated tableaux to monomials, the diagonal 
action of $\sg_{dn}$ is just permuting the entries in the correlated tableau. In 
particular, if $i$ and $j$ are in the same entry of $C$, then switching 
$i$ and $j$ leaves $C$ unchanged, as two $-1$ are produced from the two monomials 
associated with $C$. 
\end{proof}

\begin{proposition}
As an $\F\sg_{dn}$ module, $\posetring(d)\otimes\pair{\posetring}(d)$ is 
isomorphic to $\oplus_{D}\rowtab(\diagpart(D))$, where $D$ runs over all $(d, 
n)$-correlated diagrams.
\end{proposition}

We also note that the well-known Robinson-Schensted-Knuth correspondence gives a 
1-to-1 correspondence between $(d, n)$-correlated diagrams and pairs of 
semistandard tableaux with the same shape $\lambda$, and the same content as 
$n^d$. Note that $\lambda$ necessarily dominates $n^d$ to satisfy the 
semistandard condition. On the other hand, such pairs of semistandard tableaux  
could be used to index those $\F\sg_{dn}$ sub-bimodules of 
$\posetring(d)\otimes\pair{\posetring}(d)$ of isomorphic type 
$\specht(\conjpart{\lambda})\otimes\specht(\conjpart{\lambda})$. This follows by 
adapting the results in \cite[Chap. 2.10]{Sagan} to the column tabloid setting. 

Recall that $\phi$ realizes $\posetring(d)\otimes\pair{\posetring}(d)$ as 
functions on $(V\otimes W)^{\oplus dn}$, as defined in 
Equation~\ref{eq:phi}. We now describe another set of vectors in $\ker(\phi)$ 
based on the $\F\sg_{dn}$ 
structure. Given a correlated tableau $C$, suppose $C$ has $\geq n^2+1$ nonempty 
entries. Note that this requires $d\geq n+1$. Fix $N:=n^2+1$ entries and from each 
entry choose a number, denoted as $I=\{i_1, \dots, i_{n^2+1}\}$. Then consider the 
following vector in $\posetring(d)\otimes\pair{\posetring}(d)$:
\begin{equation}\label{eq:alt}
\alt(C, I)=\sum_{\sigma\in \sg_I}\sgn(\sigma) C^{\sigma},
\end{equation}
where $\sigma$ acts on $C$ as in the diagonal action of $\sg_{dn}$. Then 
$\phi(\alt(C, I))$ is a zero function on $(V\otimes W)^{\oplus dn}$, since it is 
alternating in $n^2+1$ copies of $V\otimes W$, an $n^2$-dimensional vector space. 
Note that construction is very natural in light of the identification of 
$\ctab(D)$ with $\rowtab(\diagpart(D))$.

\section{Towards proving degree bounds}\label{sec:deg}

\paragraph{Reduction to the multilinear case.} Over characteristic $0$ fields, we 
can also reduce the degree bound problem to the multilinear case as well. 
Consider a graded invariant ring $R$, for which we want to 
prove that $R$ is generated as a ring by $R_{\leq b}=\{f\in R\mid \deg(f)\leq 
b\}$. If $R$ is generated by $R_{\leq b}$, then any homogeneous multilinear 
invariant $f$ of degree $b'>b$ is equal to $\sum_i\prod_{j} f_{i,j}$ where 
$b\geq \deg(f_{i,j})\geq 1$. On the other hand, suppose each multilinear invariant 
of degree $b'>b$ can be written as such. 
Take any homogeneous $g\in R$ of degree $b'>b$, fully polarize it to get a 
multilinear $f$, 
which then by assumption can be written as $\sum_i\prod_{j} f_{i,j}$ where 
$b\geq \deg(f_{i,j})\geq 1$. Then restitute back to get $b'! g$ on one hand, 
and on the other hand $\sum_i\prod_{j} g_{i,j}$, where $g_{i,j}$ are obtained via 
restitution to $f_{i,j}$. So $b\geq \deg(g_{i,j})=\deg(f_{i,j})\geq 1$, which 
implies that $R$ is generated by $R_{\leq b}$.

\paragraph{Some reducible polynomials in 
$\posetring(d)\otimes\pair{\posetring}(d)$.} Recall that 
$\posetring(d)\otimes\pair{\posetring}(d)$ is spanned by $S\otimes \pair{T}$, 
where $S$ and $\pair{T}$ are monomials; or equivalently, by correlated tableaux 
$C$. We now view $C$ as a bipartite graph on $(L\cup R, E)$ allowing multiple 
edges, where $L=R=[d]$, and there are $k$ edges between $i$ and $j$ if 
$C(i,j)=k$. 

We shall prove that $\phi(C)$ is reducible if and only if this graph is disconnected. To show this, we recall the matrix-theoretic 
interpretation of $\phi(C)$ as in \cite{ANK07}. Let $A_1, \dots, A_s$ be $s$ 
variable matrices of size $s\times s$, and $y_1, \dots, y_s$ are newly introduced 
variables. The \emph{mixed discriminant} $\mdisc$ is a polynomial in $s^3$ 
variables in $A_i$'s, defined as follows. Given $\sigma\in\sg_s$, let $A^\sigma$ 
be the matrix whose $i$th column is the $i$th column of $A_{\sigma(i)}$. Then 
$\mdisc(A_1, \dots, A_s)=\sum_{\sigma\in\sg_s}\det(A^\sigma)$.

Given a correlated tableau $C$, we form $dn$ matrices of size 
$dn\times dn$ $Y_1, \dots, Y_{dn}$ as follows. For every $i$, $Y_i$ is viewed a 
$d\times 
d$ block matrix, where each block is of size $n\times n$. If $i$ appears in the 
$(j, k)$th position of $C$, then $Y_i$ has $X_i$ at the $(j, k)$ block, and $0$ 
everywhere else. 
In particular note that by construction $\mdisc(Y_1, \dots, Y_{dn})$ is a 
multilinear polynomial in $X_i$'s.
\begin{theorem}[{\cite{ANK07}}]
$\phi(C)$ is equal to $\mdisc(Y_1, \dots, Y_{dn})$.
\end{theorem}
\begin{proposition}
$\phi(C)$ is a reducible polynomial if and only if $C$ is a disconnected graph. 
\end{proposition}
\begin{proof}
If $C$ is disconnected, then $\phi(C)$ is a product of at least two 
multilinear polynomials of smaller degree, each of which corresponds to a 
connected component. On the other hand, if $C$ is connected, let us assume that 
$\phi(C)$ can be written as $f\cdot g$. Without loss of generality assume $f$ is 
not constant, and let us show that $g$ is constant. As $f$ is not constant, some 
variable $X_i(a, b)$ appears in $f$. Suppose a variable 
$X_p(c, d)$ appearing in $\phi(C)$ does not appear with $X_i(a, b)$ in any 
monomial. Since $\phi(C)$ is multilinear and $X_j(c,d)$ appears in $\phi(C)$, the 
only way this can happen  
is if $X_j(c, d)$ is also in $f$. We say $X_j(c, d)$ connects to $X_i(a, b)$ if 
this happens. Note that if $X_k(u, v)$ connects to $X_j(c, d)$ and $X_j(c, d)$ 
connects to $X_i(a, b)$, then $X_k(u, v)$ is in $f$ as well. It is not 
hard to see from mixed discriminant perspective that: (1) every variable in $X_i$ 
connects to $X_i(a, b)$; (2) if $X_j$ is in the same row or column as $X_i$ in $C$,
then every variable in $X_j$ connects to some variable in $X_i$. Since $C$ is a 
connected graph, using (1) and (2) iteratively we put every variable in $f$, which 
implies that $g$ is a constant. 
\end{proof}

Let $\ctabd(d)$ be the span of disconnected correlated tableaux in 
$\posetring(d)\otimes \pair{\posetring}(d)$, and $ \ctabs(d)$ the span of connected 
correlated tableaux.

\paragraph{Using the second fundamental theorem.} 
A degree bound of, say, degree $bn$, would involve showing that every invariant of degree $dn$, $d >b$, is equal to a sum of products of invariants of smaller degree, modulo relations in $\phi$. 
Algorithmically we would start with an invariant of degree $d$, and express it as a sum of products of invariants of smaller degree plus
something in the kernel $\phi$. If in this new expression there still remain invariants of degree bigger than $b$, we write those as sums of products of invariants of smaller degree plus elements in the kernel $\phi$. We continue, till all invariants involved in the final expression have degree less than or equal to $bn$.

We are then led to examine the following situation -- if $B$ is a set of 
polynomials spanning $R(n, dn)$, is every irreducible 
polynomial in $B$ equal to a sum of reducible ones in $B$ plus something in the kernel of $\phi$?
If this happens beyond a certain degree $bn$, we get an upper bound of $bn$, on the degree in which the invariant ring is generated. Taking $B$ as $\phi(C)$, where $C$ runs over all correlated tableaus, we summarise the above argument in the following

\begin{proposition}\label{prop:deg}
$R(n, m)$ is generated by $R(n, m)_{\leq bn}$ if and only if for every $d>b$, 
$\ctabd(d)\cup \ker(\phi)$ spans $\posetring(d)\otimes\pair{\posetring}(d)$.
\end{proposition}

Using the bilinear form $\beta$ introduced in Fact~\ref{fact:bimodule} (3), we 
have the following.
\begin{corollary}
$R(n, m)$ is generated by $R(n, m)_{\leq bn}$ if and only if for every $d>b$, 
$\ctabc(d)\cap \big(\ipo(d)\otimes\ipo(d)\big)=0$.
\end{corollary}

We now use some results developed so far to lower bound the degree in which the invariant ring 
can be generated as an algebra.

\begin{proposition}\label{prop:lower}
If $R(n,m)$ is generated in degree $bn$, then $b = \Omega(n^{1/2})$.
\end{proposition}
\begin{proof}
We use proposition ~\ref{prop:deg}. First observe that for $\ctabd(d)\cup\ker(\phi)$ to span 
$\posetring(d)\otimes\pair{\posetring}(d)$ it is necessary that 
$\dim(\ctabd(d))+\dim(\ker(\phi))\geq 
\dim(\posetring(d)\otimes\pair{\posetring(d)})$. We show that when $d < o(n^{1/2})$ this does not happen, by doing a dimension count.
To compute the dimensions proceed as follows.
\begin{enumerate}
\item 
$\dim(\posetring(d)\otimes\pair{\posetring(d)})=D_1:=\big(\frac{(nd)!}{(n!)^d}\big)^2$;
\item By the hook length formula, 
$\dim(\ipo(d)\otimes\ipo(d))=D_2:=\big(\frac{(nd)!}{(1\cdot\ldots\cdot n)\cdot 
(2\cdot \ldots \cdot (n+1))\cdot \ldots \cdot (d\cdot \ldots\cdot 
(d+n-1))}\big)^2$;
\item $\dim(\ker(\phi))=D_1-D_2$;
\item A lower bound on $\dim(\ctabd(d))$ is 
$\big(\frac{(n(d-1))!}{(n!)^{d-1}}\big)^2$. This counts only those 
disconnected graphs with $C(1, 1)=\{1, \dots, n\}$. 
\item An upper bound on $\dim(\ctabd(d))$ is 
$$
\sum_{k=1}^{d-1} \binom{d}{k}^2 \cdot \binom{dn}{kn}\cdot 
\big(\frac{(kn)!}{(n!)^k}\big)^2 \cdot 
\big(\frac{((d-k)n)!}{(n!)^{d-k}}\big)^2,
$$
where $\binom{d}{k}^2$ fixes $k$ rows and columns as a connected component $C$, 
$\binom{dn}{kn}$ chooses $kn$ numbers as the labels within $C$, and the rest terms 
count the number of ways filling in $C$ and the other component.
\end{enumerate}
Using (3) and (5), we prove that when $d=o(\sqrt{n})$ and $n$ is large enough, 
then $\dim(\ctabd(d)) < 
\dim(\posetring(d)\otimes\pair{\posetring(d)})-\dim(\ker(\phi))$, which implies 
that the degree bound is necessarily $\Omega(n^{3/2})$. For this we need to 
verify that 
$$
\sum_{k=1}^{d-1} \binom{d}{k}^2 \cdot \binom{dn}{kn}\cdot 
\big(\frac{(kn)!}{(n!)^k}\big)^2 \cdot 
\big(\frac{((d-k)n)!}{(n!)^{d-k}}\big)^2 < \big(\frac{(nd)!}{(1\cdot\ldots\cdot 
n)\cdot 
(2\cdot \ldots \cdot (n+1))\cdot \ldots \cdot (d\cdot \ldots\cdot 
(d+n-1))}\big)^2
$$
for $d = o(\sqrt{n})$. It is easy to see that this is equivalent to showing that  
$$
\sum_{k=1}^{d-1} \binom{d}{k}^2 \cdot \frac{(kn)! ((d-k)n)!}{(dn)!} 
< 
\big(\frac{(n!)^d}{(1\cdot\ldots\cdot 
n)\cdot 
(2\cdot \ldots \cdot (n+1))\cdot \ldots \cdot (d\cdot \ldots\cdot 
(d+n-1))}\big)^2.
$$
Upper bounding the left hand side by $2^{2d} \cdot \frac{n!((d-1)n)!}{(dn)!}$ and 
rearranging the terms we need to show that 
$$
\frac{(dn)!}{n!((d-1)n)!} > 
\big(2^d\cdot \frac{(1\cdot\ldots\cdot 
n)\cdot 
(2\cdot \ldots \cdot (n+1))\cdot \ldots \cdot (d\cdot \ldots\cdot 
(d+n-1))}{(n!)^d}\big)^2.
$$
The left hand side is at least $\frac{(dn-n+1)\cdot \ldots \cdot 
(dn)}{n^n}$ and the right hand side is at most  
$$
[2^d\cdot (n+1)^{d-1}\cdot (n+2)^{d-2}\cdot \ldots \cdot (n+d-1)]^2,
$$
which is in turn at most 
$$
[2^d\cdot (n+d-1)^{\binom{d}{2}}]^2.
$$
It is easy to verify now that when $d=o(\sqrt{n})$, $(dn-n+1)\cdot \ldots 
\cdot (dn)$ is asymptotically larger than $n^n\cdot [2^d\cdot 
(n+d-1)^{\binom{d}{2}}]^2$.
\end{proof}
On the other hand, using (4) and (5) it is easy to deduce that when 
$d=\omega(\sqrt{n})$ then $\dim(\ctabd(d)) > D_2$ asymptotically in $n$, and when 
$d>4n$ then $\dim(\ctabd(d)) > D_2$ unconditionally. This means that once $d$ is 
$\Omega(n)$, then $\dim(\ctabd(d))+\dim(\ker(\phi))\geq 
\dim(\posetring(d)\otimes\pair{\posetring(d)})$, clearing a bottleneck to prove a 
polynomial degree bound for $R(n, m)$.

\paragraph{On connected graphs without multiple edges.} Recall that by 
Proposition~\ref{prop:deg} we need to consider whether $\ctabd(d)\cup \ker(\phi)$ 
spans $\posetring(d)\otimes\pair{\posetring}(d)$ for large enough $d$. In this 
section we prove an analogous, but weaker result. 

Let $\ctabs(d)$ be the span of those correlated tableaux that are connected, and 
simple; that is, with no multiple edges. $\ctabs(d)$ is clearly a subspace of 
$\ctabc(d)$. Let 
$\ctabm(d)$ be the span of the tableaux that are not in $\ctabs(d)$. Namely those 
tableaux spanning $\ctabm(d)$ are either disconnected, or connected and 
with at least one multiple edge. In particular they include any tableau with at 
least one multiple edge. 

In the proof of the following proposition, besides the relations 
$\kpo(d)\otimes \pair{\posetring}(d)$ and 
$\posetring(d)\otimes\pair{\kpo}(d)$, we crucially use those ones as described in 
Equation~\ref{eq:alt}. We then refer to relations in $\kpo(d)\otimes 
\pair{\posetring}(d)$ and 
$\posetring(d)\otimes\pair{\kpo}(d)$ as the first type, and relations from  
Equation~\ref{eq:alt} as second type realtions.

\begin{proposition}\label{prop:simple_span}
Let $d\geq n+1$. Then $\posetring(d)\otimes\pair{\posetring}(d)$ is spanned by 
$\ctabm(d)\cup\ker(\phi)$. 
\end{proposition}
Note that if we define $\mathfrak{U}(d)$ the span of correlated tableaux 
whose graphs have at least one multiple edge, then this proposition is just to say 
that when $d>n$, $\ctabd(d)\cup \ker(\phi)\cup \mathfrak{U}(d)$ spans 
$\posetring(d)\otimes\pair{\posetring}(d)$. 
\begin{proof}
Take any correlated tableau $C\in\ctabs(d)$; namely the bipartite graph associated 
to $C$ is connected and without 
multiple edges. We shall show that $C$ can be written as a linear combination of 
tableaux with multiple edges via the help of $\ker(\phi)$. 

To do that let $G$ be the bipartite graph of $C$, and $T$ be the set of edges 
of $G$. ($T$ is of course labelled by $[nd]$, but we make $T$ explicitly for 
clarity.) Note that $|T|=dn\geq n^2+1$. For some $\sigma\in \sg_T$, suppose 
$\sigma=\tau_1\tau_2\dots \tau_k$ where $\tau_i$ is a transposition switching $e$ 
and $e'$ in $T$. As $G$ is connected, for each pair of edges $e, e'\in T$, $e$ is 
connected to $e'$ by some path in $G$. Thus $\tau_i$ can be further decomposed as 
a product of transpositions consisting of edges along that path connecting $e$ and 
$e'$. That is for each $i$ we have $\tau_i=\pi_{i,1}\pi_{i, 2}\dots\pi_{i, j_i}$, 
where $\pi_{i, j}$ switches two edges sharing a common vertex. 

Consider now $\pi_{1, 1}$, switching adjacent edges $e$ and $e'$, and suppose the 
position of $e$ is $(s, t_1)$ and $e'$ is $(s, t_2)$, $s, t_1, t_2\in[d]$ and $t_1 
< t_2$. That is, we assume $e$ and $e'$ are in the same row as in the correlated 
tableau. By the relations of 
the first type, we move $e$ to the $t_2$th column along the $s$th row and get 
$C=-C^{\pi_{1, 1}}+\sum_{i}D_i$, where $D_i\in\ctabm(n, d)$. We 
explain the terms on the RHS: the $-1$ sign before $C^{\pi_{1, 1}}$ is 
because when we switch $e$ and $e'$, the order of reading the $r$th row changed: 
from first reading $e$ and then $e'$ to first $e'$ and then $e$. Note that 
$C^{\pi_{1, 1}}$ is of the same shape as $C$. $D_i$'s are in 
$\ctabm(n, d)$ because as long as $e$ and $e'$ are not switched, we would have $e$ 
and $e'$ both at position $(s, t_2)$. (There might be some $-1$ before 
$D_i$'s too, but this can be neglected.) The example given after this proof illustrates this
calculation.

We then apply other $\pi_{i, j}$'s to $C^{\pi_{1, 1}}$ sequentially; each 
application of $\pi_{i, j}$ would yield a bunch of $D_i$'s in $\ctabm(n, d)$. At 
last we shall get $C=\sgn(\sigma)C^\sigma+\sum_{i}D_i$ where 
$D_i\in\ctabm(n, d)$. For every $\sigma\in \sg_T$ such an equation can be derived 
(when $\sigma=\mathrm{id}$ use $C=C$), and we have $|T|! 
C=\sum_{\sigma\in \sg_T}\sgn(\sigma)C^\sigma+\sum_{i}D_i$. 
Plugging the second type of relations in, and noting that $|T|>n^2$, 
$C=\frac{1}{|T|!} (\sum_{i}D_i)$ where $D_i\in\ctabm$ completing the proof.
\end{proof}

\begin{example}
In this example
$$
S=\begin{vmatrix}
1 & 2  \\
  3 & 4 \\
  5 & 6
\end{vmatrix}, \pair{T}=
\begin{vmatrix}
\pair{1} & \pair{3}\\
  \pair{2} & \pair{5}\\
  \pair{4} & \pair{6}
\end{vmatrix}.
$$ 
The corresponding graph is
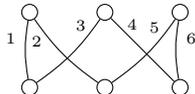
\[\begin{tikzpicture}
\vertex (u1) at (0,1)  {};
\vertex (u2) at (1,1)  {};
\vertex (u3) at (2,1)  {};
\vertex (v1) at (0,0)  {};
\vertex (v2) at (1,0)  {};
\vertex (v3) at (2,0)  {};
\path
(u1) edge[bend right=10]node[pos=0.3,left]{\tiny{$1$}} (v1)
(u1) edge[bend right=10]node[pos=0.3,left]{\tiny{$2$}} (v2)
(u2) edge[bend left=10]node[pos=0.07,left]{\tiny{$3$}} (v1)
(u2) edge node[pos=0.1,right]{\tiny{$4$}} (v3)
(u3) edge[bend left=10]node[pos=0.1,left]{\tiny{$5$}} (v2)
(u3) edge[bend right=10]node[pos=0.3,right]{\tiny{$6$}} (v3)
;
\end{tikzpicture}\]

The kernel has the following relation coming from the right monomial.
$$ \begin{vmatrix}
\pair{1} & \pair{3}\\
  \pair{2} & \pair{5} \\
  \pair{4} & \pair{6}
\end{vmatrix} =
\begin{vmatrix}
\pair{2} & \pair{3}\\
  \pair{1} & \pair{5}\\
  \pair{4} & \pair{6}
\end{vmatrix} +
\begin{vmatrix}
\pair{1} & \pair{2}\\
  \pair{3} & \pair{5}\\
  \pair{4} & \pair{6}
\end{vmatrix}
$$
Multiplying this by $S$ and recalling our convention on associating monomials to correlated tableaus, we
have 
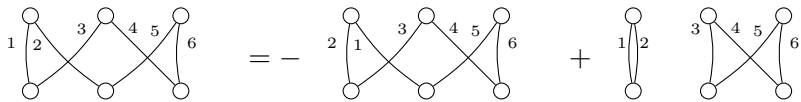
\begin{eqnarray*}
\begin{tabular}{c}
\begin{tikzpicture}
\vertex (u1) at (0,1)  {};
\vertex (u2) at (1,1)  {};
\vertex (u3) at (2,1)  {};
\vertex (v1) at (0,0)  {};
\vertex (v2) at (1,0)  {};
\vertex (v3) at (2,0)  {};
\path
(u1) edge[bend right=10]node[pos=0.3,left]{\tiny{$1$}} (v1)
(u1) edge[bend right=10]node[pos=0.3,left]{\tiny{$2$}} (v2)
(u2) edge[bend left=10]node[pos=0.07,left]{\tiny{$3$}} (v1)
(u2) edge node[pos=0.1,right]{\tiny{$4$}} (v3)
(u3) edge[bend left=10]node[pos=0.1,left]{\tiny{$5$}} (v2)
(u3) edge[bend right=10]node[pos=0.3,right]{\tiny{$6$}} (v3)
;
\end{tikzpicture}
\end{tabular}
&=- 
\begin{tabular}{c}
\begin{tikzpicture}
\vertex (u1) at (0,1)  {};
\vertex (u2) at (1,1)  {};
\vertex (u3) at (2,1)  {};
\vertex (v1) at (0,0)  {};
\vertex (v2) at (1,0)  {};
\vertex (v3) at (2,0)  {};
\path
(u1) edge[bend right=10]node[pos=0.3,left]{\tiny{$2$}} (v1)
(u1) edge[bend right=10]node[pos=0.3,left]{\tiny{$1$}} (v2)
(u2) edge[bend left=10]node[pos=0.07,left]{\tiny{$3$}} (v1)
(u2) edge node[pos=0.1,right]{\tiny{$4$}} (v3)
(u3) edge[bend left=10]node[pos=0.1,left]{\tiny{$5$}} (v2)
(u3) edge[bend right=10]node[pos=0.3,right]{\tiny{$6$}} (v3)
;
\end{tikzpicture} 
\end{tabular}
& +
\begin{tabular}{c}
\begin{tikzpicture}
\vertex (u1) at (0,1)  {};
\vertex (u2) at (1,1)  {};
\vertex (u3) at (2,1)  {};
\vertex (v1) at (0,0)  {};
\vertex (v2) at (1,0)  {};
\vertex (v3) at (2,0)  {};
\path
(u1) edge[bend left=10]node[pos=0.3,left]{\tiny{$1$}} (v1)
(u1) edge[bend right=10]node[pos=0.3,right]{\tiny{$2$}} (v1)
(u2) edge[bend left=10]node[pos=0.07,left]{\tiny{$3$}} (v2)
(u2) edge node[pos=0.1,right]{\tiny{$4$}} (v3)
(u3) edge[bend left=10]node[pos=0.1,left]{\tiny{$5$}} (v2)
(u3) edge[bend right=10]node[pos=0.3,right]{\tiny{$6$}} (v3)
;
\end{tikzpicture} 
\end{tabular}\\
\end{eqnarray*}

\end{example}




To illustrate the utility of Proposition~\ref{prop:simple_span} developed so far, 
we consider matrix semi-invariants for $2\times 2$ matrices. 
For $2\times 2$ matrices, note that as soon as there exists a multiple edge, a 
2-regular bipartite graph is disconnected. Therefore, an immediate application of 
Proposition~\ref{prop:simple_span} gives the following known result. 
\begin{theorem}\label{thm:2times2}
The matrix semi-invariants of $2\times 2$ matrices are generated by those of 
degree $\leq 4$. 
\end{theorem}
For $2\times 2$ matrices, Domokos presented an explicit generating set, and from 
this description he deduced that $\beta= 4$ \cite{Domokos00_2}. Therefore our 
bound is tight in this case. Also note that our result is obtained without 
computing a single invariant. 

\paragraph{Acknowledgement.} We would like to thank M{\'a}ty{\'a}s Domokos, Bharat 
Adsul, Ketan Mulmuley, Partha Mukhopadhyay and K. N. Raghavan for discussions 
related to this work. 
Part of the work was done when Youming was visiting the Simons Institute for the 
program Algorithms and Complexity in Algebraic Geometry.

\bibliographystyle{alpha}
\bibliography{references}

\begin{thebibliography}{Dom00b}

\bibitem[ANS07]{ANK07}
B.~Adsul, S.~Nayak, and K.~V. Subrahmanyam.
\newblock A geometric approach to the {Kronecker problem II: rectangular
  shapes, invariants of matrices and the Artin{\textendash}Procesi theorem}.
\newblock preprint, 2007.

\bibitem[BD06]{BD06}
M.~B{\"u}rgin and J.~Draisma.
\newblock The {Hilbert} null-cone on tuples of matrices and bilinear forms.
\newblock {\em Mathematische Zeitschrift}, 254(4):785--809, 2006.

\bibitem[CIW97]{CIW}
Ajeh~M Cohen, G{\'a}bor Ivanyos, and David~B Wales.
\newblock Finding the radical of an algebra of linear transformations.
\newblock {\em Journal of Pure and Applied Algebra}, 117:177--193, 1997.

\bibitem[DD12]{DD12}
M.~Domokos and V.~Drensky.
\newblock Defining relation for semi-invariants of three by three matrix
  triples.
\newblock {\em Journal of Pure and Applied Algebra}, 216(10):2098--2105, 2012.

\bibitem[Der01]{derksen_bound}
Harm Derksen.
\newblock Polynomial bounds for rings of invariants.
\newblock {\em Proceedings of the American Mathematical Society},
  129(4):955--964, 2001.

\bibitem[DKZ02]{Domokos02b}
M.~Domokos, S.~G. Kuzmin, and A.~N. Zubkov.
\newblock Rings of matrix invariants in positive characteristic.
\newblock {\em Journal of Pure and Applied Algebra}, 176(1):61--80, 2002.

\bibitem[Dom00a]{Domokos00_2}
M.~Domokos.
\newblock Poincar{\'e} series of semi-invariants of 2$\times$ 2 matrices.
\newblock {\em Linear Algebra and its Applications}, 310(1):183--194, 2000.

\bibitem[Dom00b]{Domokos00}
M.~Domokos.
\newblock Relative invariants of 3$\times$ 3 matrix triples.
\newblock {\em Linear and Multilinear Algebra}, 47(2):175--190, 2000.

\bibitem[Dom02]{Domokos02a}
M.~Domokos.
\newblock Finite generating system of matrix invariants.
\newblock {\em Math. Pannon}, 13(2):175--181, 2002.

\bibitem[Dom03]{Dom_description}
M.~Domokos.
\newblock {\em Matrix invariants and the failure of Weyl's theorem}, volume 235
  of {\em Lecture Notes in Pure and Applied Mathematics}, page
  215{\textendash}236.
\newblock Dekker, New York, 2003.

\bibitem[Don92]{Donkin92}
Stephen Donkin.
\newblock Invariants of several matrices.
\newblock {\em Inventiones mathematicae}, 110(1):389--401, 1992.

\bibitem[Don93]{Donkin93}
Stephen Donkin.
\newblock Invariant functions on matrices.
\newblock {\em Mathematical Proceedings of the Cambridge Philosophical
  Society}, 113:23--43, 1 1993.

\bibitem[DW00]{DW00}
Harm Derksen and Jerzy Weyman.
\newblock Semi-invariants of quivers and saturation for littlewood-richardson
  coefficients.
\newblock {\em Journal of the American Mathematical Society}, 13(3):467--479,
  2000.

\bibitem[DZ01]{DZ01}
M.~Domokos and A.~N. Zubkov.
\newblock Semi-invariants of quivers as determinants.
\newblock {\em Transformation groups}, 6(1):9--24, 2001.

\bibitem[For86]{formanek_gen}
Edward Formanek.
\newblock Generating the ring of matrix invariants.
\newblock In Freddy M.~J. van Oystaeyen, editor, {\em Ring Theory}, volume 1197
  of {\em Lecture Notes in Mathematics}, pages 73--82. Springer Berlin
  Heidelberg, 1986.

\bibitem[Ful97]{Fulton}
W.~Fulton.
\newblock {\em Young Tableaux: With Applications to Representation Theory and
  Geometry}.
\newblock London Mathematical Society Student Texts. Cambridge University
  Press, 1997.

\bibitem[Gur04]{Gurvits}
Leonid Gurvits.
\newblock Classical complexity and quantum entanglement.
\newblock {\em J. Comput. Syst. Sci.}, 69(3):448--484, 2004.

\bibitem[HW14]{HW14}
Pavel Hrube\v{s} and Avi Wigderson.
\newblock Non-commutative arithmetic circuits with division.
\newblock In {\em Innovations in Theoretical Computer Science, ITCS'14,
  Princeton, NJ, USA, January 12-14, 2014}, pages 49--66, 2014.
\newblock Full version at
  \url{http://www.math.ias.edu/~avi/PUBLICATIONS/HrubesWiDec2013.pdf}.

\bibitem[IQS15]{IQS1}
G\'abor Ivanyos, Youming Qiao, and K.~V. Subrahmanyam.
\newblock Non-commutative {E}dmonds' problem and matrix semi-invariants.
\newblock {\em CoRR}, abs/1508.00690, 2015.

\bibitem[Pro76]{Pro76}
C.~Procesi.
\newblock The invariant theory of $n\times n$ matrices.
\newblock {\em Advances in Mathematics}, 19(3):306--381, 1976.

\bibitem[Raz74]{Raz74}
Ju.~P. Razmyslov.
\newblock Trace identities of full matrix algebras over a field of
  characteristic zero.
\newblock {\em Mathematics of the USSR-Izvestiya}, 8(4):727, 1974.
\newblock English translation available at
  \url{http://iopscience.iop.org/0025-5726/8/4/A01}.

\bibitem[Sag01]{Sagan}
B.E. Sagan.
\newblock {\em The Symmetric Group: Representations, Combinatorial Algorithms,
  and Symmetric Functions}.
\newblock Graduate Texts in Mathematics. Springer, 2001.

\bibitem[SVdB01]{SV01}
Aidan Schofield and Michel Van~den Bergh.
\newblock Semi-invariants of quivers for arbitrary dimension vectors.
\newblock {\em Indagationes Mathematicae}, 12(1):125--138, 2001.

\bibitem[Wey97]{Wey97}
H.~Weyl.
\newblock {\em The classical groups: their invariants and representations},
  volume~1.
\newblock Princeton University Press, 1997.

\end{thebibliography}



\end{document}